\documentclass[preprint,12pt]{elsarticle}
\journal{Artificial Intelligence}

\makeatletter
\def\ps@pprintTitle{%
  \let\@oddhead\@empty
  \let\@evenhead\@empty
  \let\@oddfoot\@empty
  \let\@evenfoot\@oddfoot
}
\makeatother
\usepackage{booktabs} % For formal tables
\usepackage{amssymb}
\usepackage{amsmath}
\usepackage{wasysym}
\usepackage{epigraph}
\usepackage{graphics}
\usepackage{setspace}
\usepackage{fancybox} 
\newtheorem{definition}{Definition}
\newtheorem{theorem}{Theorem}
\newtheorem{lemma}{Lemma}
\newtheorem{claim}{Claim}

\newenvironment{proof}{\noindent{\sf Proof.}}{\hfill $\boxtimes\hspace{2mm}$\linebreak}
\renewcommand{\phi}{\varphi}
\renewcommand{\epsilon}{\varepsilon}

\renewcommand{\phi}{\varphi}
\renewcommand{\epsilon}{\varepsilon}

\newcommand{\K}{{\sf K}}

\newsavebox{\diamonddotsavebox}
\sbox{\diamonddotsavebox}{$\Diamond$\hspace{-1.8mm}\raisebox{0.3mm}{$\cdot$}\hspace{1mm}}

\begin{document}
\begin{frontmatter}
\title{Intelligence in Strategic Games}
%\title{Knowing How to Accomplish}
%\titlenote{Produces the permission block, and
%  copyright information}
%\subtitle{}
%\subtitlenote{The full version of the author's guide is available as
%  \texttt{acmart.pdf} document}

\author{Pavel Naumov}
\address{Claremont McKenna College, Claremont, California, USA}
\ead{pgn2@cornell.edu}
\author{Yuan Yuan}
\address{Vassar College, Poughkeepsie, New York, USA}
\ead{yyuan@vassar.edu}

\begin{abstract}
The article considers strategies of coalitions that are based on intelligence information about moves of some of the other agents. The main technical result is a sound and complete logical system that describes the interplay between coalition power modality with intelligence and distributed knowledge modality in games with imperfect information.
\end{abstract}

\end{frontmatter}
\section{Introduction}

The Battle of the Atlantic was a classical example of the matching pennies game. British (and American) admirals were choosing routes of the allied convoys and Germans picked routes of their U-boats. If their trajectories crossed, the Germans scored a win, if not the allies did. Neither of the players appeared to have a {\em strategy} that would guarantee victory. 

The truth, however, was that during most of the battle one of the sides had exactly such a strategy. 
First, it was British who broke German Enigma cipher in summer 1941.
Although Germans did not know about British success, they changed codebook and added fourth wheel to Enigma in February 1942 thus preventing British from decoding German messages. 
The very next month, in March 1942, German navy cryptography unit, B-Dienst,  broke allied code and got access to convoy route information. 
Germans lost their ability to read allied communication in December 1942 due to a routine change in the allied codebook.
The same month, British were able to read German communication as a result of capturing codebook from a U-boat in Mediterranean.
In March 1943, Germans changed codebook again and, unknowingly, disabled British ability to read German messages. Simultaneous, Germans caught up and started to decipher British transmissions again~\cite{b02ijnh,s03}.

At almost any moment during these two years one of the sides was able to read the communications of the other side. However, neither of them was able to figure out that their own code is insecure because the two sides never have been able to read each other messages at the same time to notice that the other side knows more than it should have known. Finally, in May 1943, with help of US Navy, British cracked German messages while Germans still were reading British. It was the first time allies understood that their code was insecure. A new convoy cipher was immediately introduced and Germans have never been able to break it again, while allies continued reading Enigma-encrypted transmissions till the end of the war~\cite{b02ijnh}.   

In this article we study coalition power in strategic games assuming that the coalition has intelligence information about moves of all or some of its opponents. We write $[C]_I\phi$ if coalition $C$ has a strategy to achieve outcome $\phi$ as long as the coalition knows what will be the move of each agent in set $I$. For example,
$$
[British]_{Germans}(\mbox{Convoy is saved}).
$$
Modality $[C]_\varnothing\phi$ is the coalitional power modality proposed by Marc Pauly~\cite{p01illc,p02}. He gave a sound and complete axiomatization of this modality in the case of perfect information strategic games. Various extensions of his logic has been studied before~\cite{g01tark,vw05ai,b07ijcai,sgvw06aamas,abvs10jal,avw09ai,b14sr,gjt13jaamas,ge18aamas}. Strategic power modality with intelligence $[a_1,\dots,a_n]_{i_1,\dots,i_k}\phi$ can be expressed in Strategy Logic~\cite{chp10ic,mmpv14tocl} as 
$$
\forall t_1\dots\forall t_k\exists s_1,\dots\exists s_n (a_1,s_1)\dots(a_n,s_n)( i_1,t_1)\dots (i_k,t_k){\sf X}\phi.
$$
The literature on the strategy logic covers model checking~\cite{bmmrv17lics}, synthesis~\cite{clm15aaai},  decidability~\cite{mmpv12concur,vpmm17lmcs}, and bisimulation~\cite{bdm18kr}. We are not aware of any completeness results for a strategy logic with quantifiers over strategies. At the same time, our approach is different from the one in Alternating-time Temporal Logic with Explicit Strategies (ATLES)~\cite{whw07tark}. There, modality $\langle\langle C\rangle\rangle_\rho$  denotes existence existence of a strategy of coalition $C$ for a fixed commitment $\rho$ of some of the other agents. Unlike ATLES, out modality $[C]_B$ denotes existence of a strategy of coalition of coalition $C$ for {\em any} commitment of coalition $B$, as long as it is known to $C$.
Goranko and Ju proposed several versions of strategic power with intelligence modality, gave formal semantics of these modalities, and discussed a matching notion of bisimulation~\cite{gj19lori}. They do not suggest any axioms for these modalities.  

An important example of intelligence in strategic games comes from {\em Stackelberg security games}. These are two-player games between a defender and an intruder. The defender is using a {\em mixed} strategy to assign available resources to targets and the intruder is using a pure strategy to attack one of the targets. The distinctive property of the security games is the assumption that the intruder knows the probabilities with which the defender assigns resources to different targets. The intruder uses this information to plan the attack that is likely to bring the most damage. In other words, it is assumed that the intruder has the intelligence about the mixed strategy deployed by the defender. Security games have been used by the U.S. Transportation Security Administration, the U.S. Federal Air Marshal Service,  the U.S. Coast Guard, and others~\cite{sfakt18ijcai}.   

Recently, logics of coalition power were generalized to imperfect information games. Unlike prefect information strategic games, the outcome of an imperfect information game might depend on the initial state of the game that could be unknown to the players. For example, consider a hypothetical setting in which an allied convoy and a German U-boat have to choose between three routes from point A to point B: route 1, route 2, or route 3, see Figure~\ref{atlantic figure}. Let us furthermore assume that it is known to both sides that one of these routes is blocked by Russian naval mines. Although the mines are located along route 1, neither allies nor Germans known this. If allies have an access to intelligence about German U-boats, then, in theory, they have a strategy to save the convoy. For example, if Germans will use route 2, then allies can use route 3. However, since allies do not know the location of Russian mines, even after they receive information about German plans, they still would not know how to save the convoy.  

\begin{figure}[ht]
\begin{center}
%\vspace{-4mm}
\scalebox{0.7}{\includegraphics{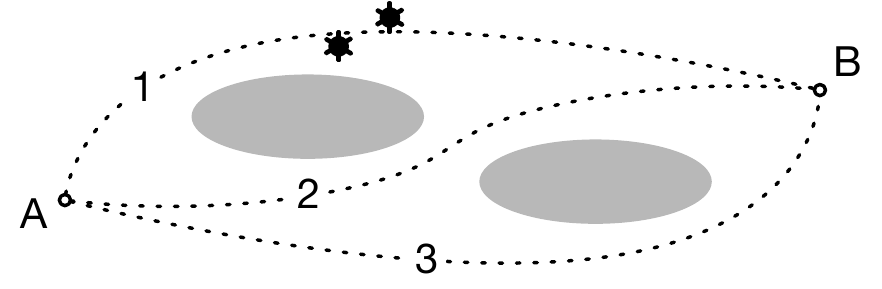}}
%\vspace{-2mm}
%\footnotesize
\caption{Three routes from point $A$ to point $B$.}\label{atlantic figure}
\vspace{-2mm}
\end{center}
%\vspace{-2mm}
\end{figure}

It has been suggested in several recent works that, in the the case of the games with imperfect information, strategic power modality in Marc Pauly logic  should be restricted to existence of {\em know-how}~\footnote{Know-how strategies were studied before under different names. While Jamroga and {\AA}gotnes talked about ``knowledge to identify and execute a strategy"~\cite{ja07jancl},  Jamroga and van der Hoek discussed ``difference between an agent knowing that he has a suitable strategy and knowing the strategy itself"~\cite{jv04fm}. Van Benthem called such strategies ``uniform"~\cite{v01ber}. Wang gave a complete axiomatization of ``knowing how" as a binary modality~\cite{w15lori,w17synthese}, but his logical system does not include the knowledge modality.} strategies~\cite{aa16jlc,nt17aamas,fhlw17ijcai,nt18aaai,nt18ai,nt18aamas}. That is, modality $[C]\phi$ should stands for ``coalition $C$ has a strategy, it knows that it has a strategy, and it knows what the strategy is''. In this article we adopt this approach to strategic power with intelligence. For example, in the imperfect information setting depicted in Figure~\ref{atlantic figure}, after receiving the intelligence report, the British have a strategy, they know that they have a strategy, but they {\em do not know} what the strategy is:
$$
\neg[British]_{Germans}(\mbox{Convoy is saved}).
$$
At the same time, since Russians presumably know the location of their mines,
$$
[\mbox{British, Russians}]_{\footnotesize\mbox{Germans}}(\mbox{Convoy is saved}).
$$

The main contribution of this article is a complete logical system that describes the interplay between the coalition power with intelligence modality $[C]_I$ and the distributed knowledge modality $\K_C$ in an imperfect information setting. The most interesting axiom of our system is a generalized version of Marc Pauly's~\cite{p01illc,p02} Cooperation axiom that connects intelligence $I$ and coalition $C$ parameters of the modality $[C]_I$. Our proof of the completeness is significantly different from the existing proofs of completeness for games with imperfect information~\cite{aa16jlc,nt17aamas,nt18aaai,nt18ai,nt18aamas}. We highlight these differences in the beginning of the Section~\ref{completeness section}. 

\section{Outline}

The rest of the article is organized as follows. In the next section we introduce the syntax and the formal semantics of our logical system. In Section~\ref{axioms section}, we list the axioms and the inference rules of the system, compare them to related axioms in the previous works and give two examples of formal proofs in our system. In the Section~\ref{soundness section} and Section~\ref{completeness section} we prove the soundness and the completeness of our logical system respectively. Section~\ref{concusion section} concludes. 

\section{Syntax and Semantics}\label{syntax and semantics section}

In this section we define the syntax and the semantics of our formal system. Throughout the article we assume a fixed set of propositional variables and a fixed set of agents $\mathcal{A}$. By a coalition we mean any finite subset of $\mathcal{A}$. Finiteness of coalitions will be important for the proof of the completeness.

\begin{definition}\label{Phi}
Let $\Phi$ be the minimal set of formulae such that
\begin{enumerate}
    \item $p\in \Phi$ for each propositional variable $p$,
    \item $\phi\to\psi, \neg\phi\in\Phi$ for all $\phi,\psi\in\Phi$,
    \item $\K_C\phi\in \Phi$ for each formula $\phi\in\Phi$ and each coalition $C\subseteq \mathcal{A}$,
    \item $[C]_B\phi\in \Phi$ for each formula $\phi\in\Phi$ and all disjoint coalitions $B,C\subseteq \mathcal{A}$.
\end{enumerate}
\end{definition}
In other words, the language of our logical system is defined by  grammar:
$$
\phi := p\;|\;\neg\phi\;|\;\phi\to\phi\;|\;\K_C\phi\;|\;[C]_B\phi. 
$$
Formula $\K_C\phi$ stands for ``coalition $C$ distributively knows $\phi$'' and formula $[C]_B\phi$ for ``coalition $C$ distributively knows strategy to achieve $\phi$ as long as it gets an intelligence on actions of coalition $B$''.

For any sets $X$ and $Y$, by $X^Y$ we mean the set of all functions from $Y$ to $X$. 

\begin{definition}\label{transition system}
A tuple $(W,\{\sim_a\}_{a\in \mathcal{A}},\Delta,M,\pi)$ is called a game if
\begin{enumerate}
    \item $W$ is a set of states,
    \item $\sim_a$ is an ``indistinguishability'' equivalence relation on set $W$ for each agent $a\in\mathcal{A}$,
    \item $\Delta$ is a nonempty set, called the ``domain of actions'', 
    \item relation $M\subseteq W\times \Delta^\mathcal{A}\times W$ is an ``aggregation mechanism'',
    \item function $\pi$ maps propositional variables to subsets of $W$.
\end{enumerate}
\end{definition}
A function $\delta$ from set $\Delta^\mathcal{A}$ is called a {\em complete action profile}. 

Figure~\ref{atlantic states figure} depicts a diagram of the Battle of the Atlantic game  with imperfect information, as described in the introduction. For the sake of simplicity, we treat British, Germans, and Russians as single agents, not groups of agents. The game has five states: 1, 2, 3, $s$, and $d$. States 1, 2, and 3 are three ``initial'' states that correspond to possible locations of Russian mines along route 1, route 2, or route 3. Neither British nor Germans can distinguish these states, which is shown in the diagram by labels on dashed lines connecting these three states. Russians know location of the mines and, thus, can distinguish these states. The other two states are ``final'' states $s$ and $d$ that describe if the convoy made it safe ($s$) or was destroyed ($d$) by either a U-boat or a mine. The designation of some states as ``initial'' and others as ``final'' is specific to the Battle of the Atlantic game. In general, our Definition~\ref{transition system} does not distinguish between such states and we allow games to take multiple consecutive transitions from one state to another.

\begin{figure}[ht]
\begin{center}
%\vspace{-4mm}
\scalebox{0.7}{\includegraphics{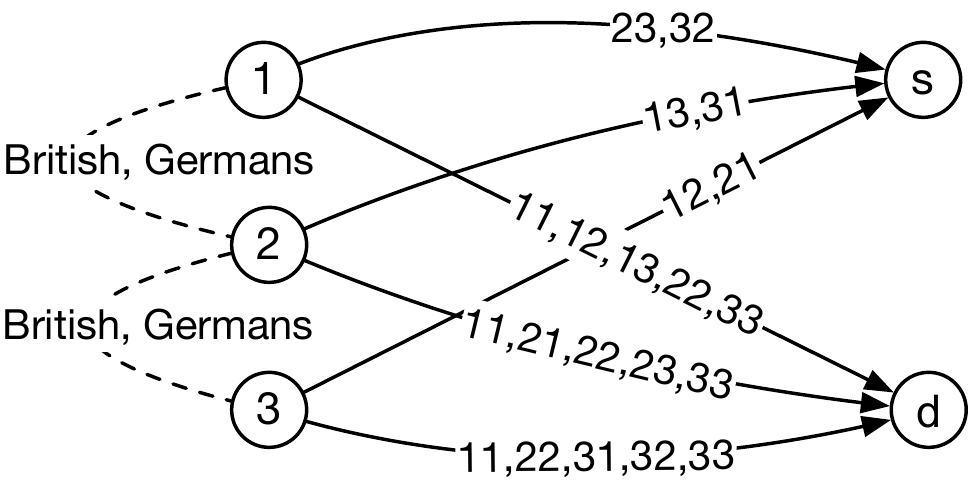}}
%\vspace{-2mm}
%\footnotesize
\caption{Battle of the Atlantic with imperfect information.}\label{atlantic states figure}
\vspace{-2mm}
\end{center}
%\vspace{-2mm}
\end{figure}

The domain of actions $\Delta$ in this game is  $\{1,2,3\}$.  For British and Germans actions represent the choice of routes that they make for their convoys and U-boats respectively. Russians are passive players in this game. Their action does not affect the outcome of the game. Technically, a complete action profile is a function $\delta$ from set $\{\mbox{British},\mbox{Germans},\mbox{Russians}\}$ into set $\{1,2,3\}$. Since, there are only three players in the Battle of the Atlantic game, it is more convenient to represent function $\delta$ by triple $bgr\in \{1,2,3\}^3$, where $b$ is the action of British, $g$ is the action of Germans, and $c$ is the action of Russians. 

The mechanism $M$ of the Battle of the Atlantic game is captured by the directed edges in Figure~\ref{atlantic states figure} labeled by complete actions profiles.
Since value $r$ in a profile $bgr$ does not effect the outcome, it is omitted on the diagram. For example, directed edge from state $1$ to state $s$ is labeled with 23 and 32. This means that the mechanism $M$ contains triples 
$(1,231,s)$,
$(1,232,s)$,
$(1,233,s)$,
$(1,321,s)$,
$(1,322,s)$, and
$(1,323,s)$.

The definition of a game that we use here is more general than the one used in the original Marc Pauly's semantics of the logic of coalition power. Namely, we assume that the mechanism is a relation, not a function. On one hand, this allows us to talk about nondeterministic games where for each initial state and each complete action profile there might be more than one outcome. On the other hand, this also allows for some combinations of the initial state and the complete action profile there to be no outcome at all. In other words, we do not exclude games in which agents might have an ability in some situations to terminate the game without reaching an outcome. If needed, such games can be excluded and an additional axiom $\neg[C]_\varnothing\bot$ be added to the logical system. The proof of the completeness will remain mostly unchanged.  We also introduce indistinguishability relation on states to capture the imperfect information. We do it in the same way as it has been done in the cited earlier previous works on the logics of coalition power with imperfect information.

\begin{definition}
For any states $w,w'\in W$ and any coalition $C$, let $w\sim_C w'$ if $w\sim_a w'$ for each agent $a\in C$.
\end{definition}
In particular, $w\sim_\varnothing w'$ for any two states of the game.

\begin{lemma}\label{sim C is equivalence relation lemma}
For any coalition $C$, relation $\sim_C$ is an equivalence relation on set $W$.\qed
\end{lemma}

By an action profile of a coalition $C$ we mean any function from set $\Delta^C$. For any two functions $f,g$, we write $f=_X g$ if $f(x)=g(x)$ for each $x\in X$. 

Next is the key definition of this article. Its part 5 gives the semantics of  modality $[C]_B$. This part uses state $w'$ to capture the fact that the strategy succeeds in each state indistinguishable by coalition $C$ from the current state $w$. In other words, the coalition $C$ {\em knows} that this strategy will succeed. Except for the addition of coalition $B$ and its action profile $\beta$, this is essentially the same definition as the one used in \cite{aa16jlc,nt17aamas,fhlw17ijcai,nt17tark,nt18ai,nt18aaai,nt18aamas}.

% \begin{figure}[ht]
% \begin{center}
% %\vspace{-4mm}
% \scalebox{0.7}{\includegraphics{Images/modality-definition.pdf}}
% %\vspace{-2mm}
% %\footnotesize
% \caption{Towards Definition~\ref{sat}.}\label{modality definition figure}
% \vspace{-4mm}
% \end{center}
% %\vspace{-2mm}
% \end{figure}

\begin{definition}\label{sat}
For any  game $(W,\{\sim_a\}_{a\in\mathcal{A}},\Delta,M,\pi)$, any state $w\in W$, and any formula $\phi\in \Phi$, let satisfiability relation $w\Vdash\phi$ be defined as follows:
\begin{enumerate}
    \item $w\Vdash p$ if $w\in \pi(p)$, where $p$ is a propositional variable,
    \item $w\Vdash\neg\phi$ if $w\nVdash\phi$,
    \item $w\Vdash\phi\to\psi$ if $w\nVdash\phi$ or $w\Vdash\psi$,
    \item $w\Vdash\K_C\phi$ if $w'\Vdash\phi$ for each $w'\in W$ such that $w\sim_C w'$, 
    \item $w\Vdash[C]_B\phi$ if for any action profile $\beta\in \Delta^B$ of coalition $B$ there is an action profile $\gamma\in \Delta^C$ of coalition $C$ such that for any complete action profile $\delta\in \Delta^\mathcal{A}$ and any states $w',u\in W$ if $\beta=_B\delta$, $\gamma=_C\delta$, $w\sim_C w'$, and $(w',\delta,u)\in M$, then $u\Vdash\phi$.
\end{enumerate}
\end{definition} 
For example, for the game depicted in Figure~\ref{atlantic states figure},
$$
1\Vdash [\mbox{British, Russians}]_{\footnotesize\mbox{Germans}}(\mbox{Convoy is saved}).
$$
Indeed, statement $(1\sim_{\mbox{\small British, Russians}}w')$ is true only for one state $w'\in W$, namely state $1$ itself. Then, for any action  profile $\beta\in\{1,2,3\}^{\{\mbox{\small Germans}\}}$ of the single-member coalition $\{\mbox{Germans}\}$ we can define  action  profile $\gamma\in\{1,2,3\}^{\{\mbox{\small British, Russians}\}}$ as, for example,
$$
\gamma(a)=
\begin{cases}
3, & \mbox{ if $a=\mbox{British}$ and $\beta(\mbox{Germans})=2$},  \\
2, & \mbox{ if $a=\mbox{British}$ and $\beta(\mbox{Germans})=3$},  \\
1, & \mbox{ if $a=\mbox{Russians}$}.
\end{cases}
$$
In other words, if profile $\beta$ assigns Germans route 2, then profile $\gamma$ assigns British route 3 and vice versa. Assignment of an action to Russians in not important. This way, no matter what Germans' action is, the British convoy will avoid both the German U-boat and the Russian mines in the game that starts from state $w'=1$. At the same time, 
$$
1\Vdash \neg [\mbox{British}]_{\footnotesize\mbox{Germans}}(\mbox{Convoy is saved}).
$$
because without Russians the British cannot distinguish states 1, 2, and 3. In other words, $(1\sim_{\mbox{\small British}}w')$ for any state $w'\in \{1,2,3\}$. Thus, for each action  profile $\beta\in\{1,2,3\}^{\{\mbox{\small Germans}\}}$ we need to have a single action profile $\gamma\in\{1,2,3\}^{\{\mbox{\small British, Russians}\}}$ that would bring the convoy to state $s$ from any of the states 1, 2, and 3. Such profile $\gamma$ does not exists because, even if the British know where Germans U-boat will be, there is no single uniform strategy to choose path that would avoid Russian mines from all three indistinguishable states 1, 2, and 3. 

\section{Axioms}\label{axioms section}

In addition to the propositional tautologies in language $\Phi$, our logical system consists of the following axioms:

\begin{enumerate}
    \item Truth: $\K_C\phi\to\phi$,
    \item Distributivity: $\K_C(\phi\to\psi)\to(\K_C\phi\to\K_C\psi)$,
    \item Negative Introspection: $\neg\K_C\phi\to\K_C\neg\K_C\phi$,
    \item Epistemic Monotonicity: $\K_C\phi\to\K_D\phi$, where $C\subseteq D$,
    \item Strategic Introspection: $[C]_B\phi\to \K_C[C]_B\phi$,
    \item Empty Coalition: $\K_\varnothing\phi\to [\varnothing]_\varnothing\phi$,
    \item Cooperation: $[C]_B(\phi\to\psi)\to([D]_{B,C}\phi\to[C,D]_B\psi)$, where sets $B,C$, and $D$ are pairwise disjoint,
    \item Intelligence Monotonicity: $[C]_B\phi\to [C]_{B'}\phi$, where $B\subseteq B'$,
    \item None to Analyze: $[\varnothing]_B\phi\to [\varnothing]_\varnothing\phi$.
\end{enumerate}
Note that in the Cooperation axiom above and often throughout the rest of the article we abbreviate $B\cup C$ as $B,C$. However, we keep writing $B\cup C$ when notation $B,C$ could be confusing.

The Truth, the Distributivity, the Negative Introspection, and the Epistemic Monotonicity axioms are the standard axioms of the epistemic logic of distributed knowledge~\cite{fhmv95}. The Strategic Introspection axiom states that if a coalition $C$ has a ``know-how'' strategy, then it knows that it has such a strategy. A version of this axiom without intelligence was first introduced in \cite{aa16jlc}.  The Empty Coalition axiom says that if statement $\phi$ is satisfied in each state of the model, then the empty coalition has a strategy to achieve it. This axiom first appeared in \cite{nt17tark}. The Cooperation axiom for strategies without intelligence:
$$
[C](\phi\to\psi)\to([D]\phi\to[C,D]\psi),
$$
where sets $C$ and $D$ are disjoint, was introduced in \cite{p01illc,p02}. This is the signature axiom that appears in all subsequent works on logics of coalition power.  The version of this axiom with intelligence is one of the key contributions of the current article. Our version states that if coalition $C$ knows how to achieve $\phi\to\psi$ assuming it has intelligence about actions of coalition $B$ and coalition $D$ knows how to achieve $\phi$ assuming it has intelligence about actions of coalitions $B$ and $C$, then coalitions $C$ and $D$ know how together they can achieve $\psi$ if they have intelligence about actions of coalition $B$. We prove soundness of this axiom in Section~\ref{soundness section}. The remaining two axioms are original to this article. The Intelligence Monotonicity axiom states if coalition $C$ has a strategy based on intelligence about actions of coalition $B$, then coalition $C$ has such strategy based on intelligence about any larger coalition. The other form of monotonicity for modality $[C]_B$, monotonicity on coalition $C$ is also true. It is not listed among our axioms because it is provable in our system, see Lemma~\ref{subscript monotonicity lemma}. The None to Analyze axiom say that if there is none to interpret the intelligence information about coalition $B$, then this intelligence might as well not exist.

We write $\vdash\phi$ if formula $\phi$ is provable from the above axioms using the Modus Ponens, the Epistemic Necessitation, and the  Strategic Necessitation inference rules:
$$
\dfrac{\phi,\;\;\;\phi\to\psi}{\psi}, 
\hspace{10mm}
\dfrac{\phi}{\K_C\phi},
\hspace{10mm}
\dfrac{\phi}{[C]_B\phi}.
$$

We write $X\vdash\phi$ if formula $\phi\in\Phi$ is provable from the theorems of our logical system and an additional set of axioms $X$ using only the Modus Ponens inference rule. Note that if set $X$ is empty, then statement $X\vdash\phi$ is equivalent to $\vdash\phi$. We say that set $X$ is consistent if $X\nvdash\bot$.

The next lemma gives an example of a formal proof in our logical system. This example will be used later in the proof of the completeness. 

\begin{lemma}\label{subscript monotonicity lemma}
$\vdash [C]_B\phi\to [C']_B\phi$, where $C\subseteq C'$.
\end{lemma}
\begin{proof}
Formula $\phi\to\phi$ is a tautology. Thus, $\vdash [{C'\setminus C}]_B(\phi\to\phi)$ by the Strategic Necessitation inference rule. Note that the following formula: 
$[C'\setminus C]_B(\phi\to\phi)\to ([C]_B\phi\to [{(C'\setminus C)\cup C}]_B\phi)$ 
is an instance of the Cooperation axiom. Thus, $\vdash [C]_B\phi\to [{(C'\setminus C)\cup C}]_B\phi$ by the Modus Ponens inference rule. Note also that $(C'\setminus C)\cup C=C'$ because of the assumption $C\subseteq C'$. Therefore, $\vdash [C]_B\phi\to [C']_B\phi$.
\end{proof}

The following lemma states the well-known Positive Introspection principle for the distributed knowledge.
\begin{lemma}\label{positive introspection lemma}
$\vdash \K_C\phi\to\K_C\K_C\phi$. 
\end{lemma}
\begin{proof}
Formula $\K_C\neg\K_C\phi\to\neg\K_C\phi$ is an instance of the Truth axiom. Thus, $\vdash \K_C\phi\to\neg\K_C\neg\K_C\phi$ by contraposition. Hence, taking into account the following instance of  the Negative Introspection axiom: $\neg\K_C\neg\K_C\phi\to\K_C\neg\K_C\neg\K_C\phi$,
we have 
\begin{equation}\label{pos intro eq 2}
\vdash \K_C\phi\to\K_C\neg\K_C\neg\K_C\phi.
\end{equation}

At the same time, $\neg\K_C\phi\to\K_C\neg\K_C\phi$ is an instance of the Negative Introspection axiom. Thus, $\vdash \neg\K_C\neg\K_C\phi\to \K_C\phi$ by the law of contrapositive in the propositional logic. Hence, by the Necessitation inference rule, 
$\vdash \K_C(\neg\K_C\neg\K_C\phi\to \K_C\phi)$. Thus, by  the Distributivity axiom and the Modus Ponens inference rule, 
$
  \vdash \K_C\neg\K_C\neg\K_C\phi\to \K_C\K_C\phi.
$
 The latter, together with statement~(\ref{pos intro eq 2}), implies the statement of the lemma by propositional reasoning.
\end{proof}

We conclude this section by stating the two standard lemmas about our deduction system. These lemmas will be used later in the proof of the completeness.

\begin{lemma}[deduction]\label{deduction lemma}
If $X,\phi\vdash\psi$, then $X\vdash\phi\to\psi$.
\end{lemma}
\begin{proof}
Since $X,\phi\vdash\psi$ refers to the provability without the use of the Epistemic Necessitation and the  Strategic Necessitation inference rules, the standard proof of the deduction lemma for propositional logic~\cite[Proposition 1.9]{m09} applies to our system as well.
\end{proof}

\begin{lemma}[Lindenbaum]\label{Lindenbaum's lemma}
Any consistent set of formulae can be extended to a maximal consistent set of formulae.
\end{lemma}
\begin{proof}
The standard proof of Lindenbaum's lemma~\cite[Proposition 2.14]{m09} applies here too.
\end{proof}

\section{Soundness}\label{soundness section}

In this section we prove the soundness of the axioms of our logical system with respect to the semantics given in Section~\ref{syntax and semantics section}. 

\begin{theorem}[soundness]
If $\vdash \phi$, then $\Vdash \phi$ for each state $w$ of each game.
\end{theorem}
As usual, the soundness of the Truth, the Distributivity, the Negative Introspection, and the Monotonicity axiom follows from the assumption that $\sim_a$ is an equivalence relation~\cite{fhmv95}. Below we prove the soundness of each of the remaining axioms as a separate lemma.

\begin{lemma}
If $w\Vdash [C]_B\phi$, then $w\Vdash\K_C[C]_B\phi$.
\end{lemma}
\begin{proof}
Consider any state $w'\in W$ such that $w\sim_C w'$. By Definition~\ref{sat}, it suffices to show that $w'\Vdash[C]_B\phi$. 
Indeed, consider any action profile $\beta\in \Delta^B$ of coalition $B$. By the same Definition~\ref{sat}, it suffices to show that there is an action profile $\gamma\in \Delta^C$ of coalition $C$ such that for any complete action profile $\delta\in \Delta^\mathcal{A}$ and any states $w'',u\in W$ if $\beta=_B\delta$, $\gamma=_C\delta$, $w'\sim_C w''$, and $(w'',\delta,u)\in M$, then $u\Vdash\phi$.

Since $w\sim_C w'$, by Lemma~\ref{sim C is equivalence relation lemma}, it suffices to show that there is an action profile $\gamma\in \Delta^C$ of coalition $C$ such that for any complete action profile $\delta\in \Delta^\mathcal{A}$ and any states $w'',u\in W$ if $\beta=_B\delta$, $\gamma=_C\delta$, $w\sim_C w''$, and $(w'',\delta,u)\in M$, then $u\Vdash\phi$. The last statement is true by Definition~\ref{sat} and the assumption  $w\Vdash[C]_B\phi$.
\end{proof}

\begin{lemma}
If $w\Vdash\K_\varnothing\phi$, then $w\Vdash[\varnothing]_\varnothing\phi$.
\end{lemma}
\begin{proof}
Let $\beta\in \Delta^\varnothing$ be an action profile of an empty coalition\footnote{Such action profile is unique, but this is not important for our proof.}. By Definition~\ref{sat}, it suffices to show that there is an action profile $\gamma\in \Delta^\varnothing$ of the empty coalition such that for any complete action profile $\delta\in \Delta^\mathcal{A}$ and all states $w',u\in W$ if $\beta=_\varnothing\delta$, $\gamma=_\varnothing\delta$, $w\sim_\varnothing w'$, and $(w',\delta,u)\in M$, then $u\Vdash\phi$. Indeed, let $\gamma=\beta$. Thus, it suffices to prove that  $u\Vdash\phi$ for each state $u\in W$. The last statement follows from the assumption $w\Vdash\K_\varnothing\phi$ by Definition~\ref{sat}.
\end{proof}

\begin{lemma}
If $w\Vdash[C]_B(\phi\to\psi)$, $w\Vdash[D]_{B,C}\phi$, and sets $B$, $C$, and $D$ are pairwise disjoint, then $w\Vdash[C,D]_{B}\psi$.
\end{lemma}
\begin{proof}
Consider any action profile $\beta\in\Delta^B$ of coalition $B$. By Definition~\ref{sat}, it suffices to show that  there is an action profile $\gamma\in \Delta^{C\cup D}$ of coalition $C\cup D$ such that for any complete action profile $\delta\in \Delta^\mathcal{A}$ and any states $w',u\in W$ if $\beta=_B\delta$, $\gamma=_{C,D}\delta$, $w\sim_{C,D} w'$, and $(w',\delta,u)\in M$, then $u\Vdash\psi$.

Assumption $w\Vdash[C]_B(\phi\to\psi)$, by Definition~\ref{sat}, implies that there is an action profile $\gamma_1\in \Delta^{C}$ of coalition $C$ such that for any complete action profile $\delta\in \Delta^\mathcal{A}$ and any states $w',u\in W$ if $\beta=_B\delta$, $\gamma_1=_{C}\delta$, $w\sim_{C} w'$, and $(w',\delta,u)\in M$, then $u\Vdash\phi\to\psi$.

Define action profile $\beta_1\in\Delta^{B\cup C}$ of coalition $B\cup C$ as follows:
$$
\beta_1(a)=
\begin{cases}
\beta(a), & \mbox{if } a\in B,\\
\gamma_1(a), & \mbox{if } a\in C. 
\end{cases}
$$
Action profile $\beta_1$ is well-defined because sets $B$ and $C$ are disjoint.

Assumption $w\Vdash[D]_{B,C}\phi$, by Definition~\ref{sat}, implies that there is an action profile $\gamma_2\in \Delta^D$ of coalition $D$ such that for any complete action profile $\delta\in \Delta^\mathcal{A}$ and any states $w',u\in W$ if $\beta_1=_{B,C}\delta$, $\gamma_2=_{D}\delta$, $w\sim_D w'$, and $(w',\delta,u)\in M$, then $u\Vdash\phi$.

Define action profile $\gamma\in\Delta^{C\cup D}$ of coalition $C\cup D$ as follows:
$$
\gamma(a)=
\begin{cases}
\gamma_1(a), & \mbox{if } a\in C,\\
\gamma_2(a), & \mbox{if } a\in D. 
\end{cases}
$$
Action profile $\gamma$ is well-defined because sets $C$ and $D$ are disjoint.

Consider any complete action profile $\delta\in \Delta^\mathcal{A}$ and any states $w',u\in W$ such that $\beta=_B\delta$, $\gamma=_{C\cup D}\delta$, $w\sim_{C\cup D} w'$, and $(w',\delta,u)\in M$. Recall from the first paragraph of this proof that it suffices to show that $u\Vdash\phi$. Note that $\beta=_B\delta$, $\gamma_1=_C\gamma=_C\delta$, $w\sim_C w'$, and $(w',\delta,u)\in M$. Thus, $u\Vdash\phi\to\psi$ by the choice of the action profile $\gamma_1$. Similarly, $\beta_1=_B\beta=_B\delta$ and $\beta_1=_C\gamma_1=_C\gamma=_C\delta$. Hence, $\beta_1=_{B\cup C}\delta$. Also, $\gamma_2=_D\gamma=_D\delta$, $w\sim_C w'$, and $(w',\delta,u)\in M$. Thus, $u\Vdash\phi$ by the choice of the action profile $\gamma_2$. Therefore, $u\Vdash\psi$ by Definition~\ref{sat} because $u\Vdash\phi\to\psi$ and $u\Vdash\phi$.
\end{proof}

\begin{lemma}
If $w\Vdash[C]_B\phi$, $B\subseteq B'$, and sets $B'$ and $C$ are disjoint, then $w\Vdash[C]_{B'}\phi$.
\end{lemma}
\begin{proof}
Consider any action profile $\beta'\in\Delta^{B'}$ of coalition $B'$. By Definition~\ref{sat}, it suffices to show that there is an action profile $\gamma\in \Delta^C$ of coalition $C$ such that for any complete action profile $\delta\in \Delta^\mathcal{A}$ and any states $w',u\in W$ if $\beta'=_B\delta$, $\gamma=_C\delta$, $w\sim_C w'$, and $(w',\delta,u)\in M$, then $u\Vdash\phi$.

Define action profile $\beta\in\Delta^B$ of coalition $B$ to be such that $\beta(a)=\beta'(a)$ for each agent $a\in B$. Action profile $\beta$ is well-defined due to the assumption $B\subseteq B'$ of the lemma. By Definition~\ref{sat}, assumption $w\Vdash[C]_B\phi$ implies that there is an action profile $\gamma\in \Delta^C$ of coalition $C$ such that for any complete action profile $\delta\in \Delta^\mathcal{A}$ and any states $w',u\in W$ if $\beta=_B\delta$, $\gamma=_C\delta$, $w\sim_C w'$, and $(w',\delta,u)\in M$, then $u\Vdash\phi$. Note that $\beta=_B\beta'$ by the choice of action profile $\beta$.  Therefore, there is an action profile $\gamma\in \Delta^C$ of coalition $C$ such that for any complete action profile $\delta\in \Delta^\mathcal{A}$ and any states $w',u\in W$ if $\beta'=_B\delta$, $\gamma=_C\delta$, $w\sim_C w'$, and $(w',\delta,u)\in M$, then $u\Vdash\phi$.
\end{proof}

\begin{lemma}
If $w\Vdash[\varnothing]_B\phi$, then $w\Vdash[\varnothing]_\varnothing\phi$.
\end{lemma}
\begin{proof}
By Definition~\ref{sat}, assumption $w\Vdash[\varnothing]_B\phi$ implies that for any action profile $\beta\in \Delta^B$ of coalition $B$ there is an action profile $\gamma\in \Delta^\varnothing$ of the empty coalition such that for any complete action profile $\delta\in \Delta^\mathcal{A}$ and any states $w',u\in W$ if $\beta=_B\delta$, $\gamma=_\varnothing\delta$, $w\sim_\varnothing w'$, and $(w',\delta,u)\in M$, then $u\Vdash\phi$.

Thus, for any action profile $\beta\in \Delta^B$ of coalition $B$, any complete action profile $\delta\in \Delta^\mathcal{A}$ and any states $w',u\in W$ if $\beta=_B\delta$ and $(w',\delta,u)\in M$, then $u\Vdash\phi$. Hence, for any complete action profile $\delta\in \Delta^\mathcal{A}$ and any states $w',u\in W$ if $(w',\delta,u)\in M$, then $u\Vdash\phi$.

Then, for any action profile $\beta\in \Delta^\varnothing$ of the empty coalition there is an action profile $\gamma\in \Delta^\varnothing$ of the empty coalition such that for any complete action profile $\delta\in \Delta^\mathcal{A}$ and any states $w',u\in W$ if $\beta=_\varnothing\delta$, $\gamma=_\varnothing\delta$, $w\sim_\varnothing w'$, and $(w',\delta,u)\in M$, then $u\Vdash\phi$.

Therefore,$w\Vdash[\varnothing]_\varnothing\phi$ by Definition~\ref{sat}.
\end{proof}

\section{Completeness}\label{completeness section}

In this section we prove the completeness of our logical system. We start this proof by fixing a maximal consistent set of formulae $X_0$ and defining the canonical game  $G(X_0)=(W,\{\sim_a\}_{a\in\mathcal{A}},\Delta,M,\pi)$. 

There are two major challenges that we need to overcome while defining the canonical model. The first of them is well-known complication related to the presence of the distributed knowledge modality in our logical system. The second is a unique challenge of specific to strategies with intelligence. 
To understand the first challenge, recall that in case of individual knowledge states are usually defined as maximal consistent sets. Two such sets are $\sim_a$-equivalent if the sets contain the same $\K_a$ formulae. Unfortunately, this construction can not be easily adapted to distributed knowledge because if two sets share $\K_a$ and $\K_b$ formulae, then they not necessarily share $\K_{a,b}$ formulae. To overcome this challenge we use ``tree'' construction in which each state is a node of a labeled tree. Nodes of the tree are labeled with maximal consistent sets and edges are labeled with coalitions. This construction has been used in logics of know-how with distributed knowledge before~\cite{nt17aamas,nt17tark,nt18ai,nt18aaai,nt18aamas}.

To understand the second challenge, let us first recall the way the canonical game is usually constructed for the logics of the coalition power. The commonly used construction defines the domain of actions to be the set of all formulae. Informally, it means that each agent ``votes'' for a formula that the agents wants to be true in the next state. Of course, not requests of agents are granted. The canonical game mechanism specifies which requests are granted and which are ignored. There also are canonical game constructions in which a voting ballot in addition to a formula must also contain some additional information that acts as a ``key'' verifying that the voting agents has certain information~\cite{nt18ai,nt18aamas}. So, it is natural to assume that in case of formula $[C]_B\phi$, coalition $C$ should vote for formula $\phi$ and provide vote of coalition $B$ as a key. This approach, however, turns out to be problematic. Indeed, in order to satisfy some other formula, say $[B]_D\psi$, vote of coalition $B$ would need to include vote of coalition $D$ as a key. Thus, it appears, that vote of $C$ would need to include vote of $D$ as well. The situation is further complicated by mutual recursion when one attempts to satisfy formulae $[C]_B\phi$ and $[B]_C\chi$ simultaneously. The solution that we propose in this article avoids this recursion. It turns out that it is not necessary for the key to contain the complete intelligence information. Namely, we assume that each agent votes for a formula and signs her vote with a random integer key. To satisfy formula $[C]_B\phi$, the mechanism will guarantee that {\em if all members of coalition $C$ vote for $\phi$ and sign with integer keys that are larger than keys of all members in coalition $B$, then $\phi$ will be true in the next state}. This idea is formalized later in Definition~\ref{canonical M}.  

\begin{definition}\label{canonical state}
Sequence $X_0,C_1,X_1,C_2,\dots,C_n,X_n$ is a state of the canonical game if
\begin{enumerate}
    \item $n\ge 0$,
    \item $X_1,\dots,X_n$ are maximal consistent sets of formulae,
    \item $C_1,\dots,C_n$ are coalitions of agents,
    \item $\{\phi\in\Phi\;|\; \K_{C_k}\phi\in X_{k-1}\}\subseteq X_k$ for each integer $k$ such that $1\le k \le n$.
\end{enumerate}
\end{definition}

We say that sequence $w=X_0,C_1,X_1,\dots,C_{n-1},X_{n-1}$ and sequence $u=X_0,C_1,X_1,\dots,C_n,X_n$ are {\em adjacent}. The adjacency relation defines an {\em undirected} labeled graph whose nodes are elements of set $W$ and whose edges are specified by the adjacency relation. The node $u$ by set $X_n$ and edge $(w,u)$ by each agent in set $C_n$, see Figure~\ref{tree figure}. Note that this graph has no cycles and thus is a tree. 
\begin{figure}[ht]
\begin{center}
%\vspace{-4mm}
\scalebox{0.7}{\includegraphics{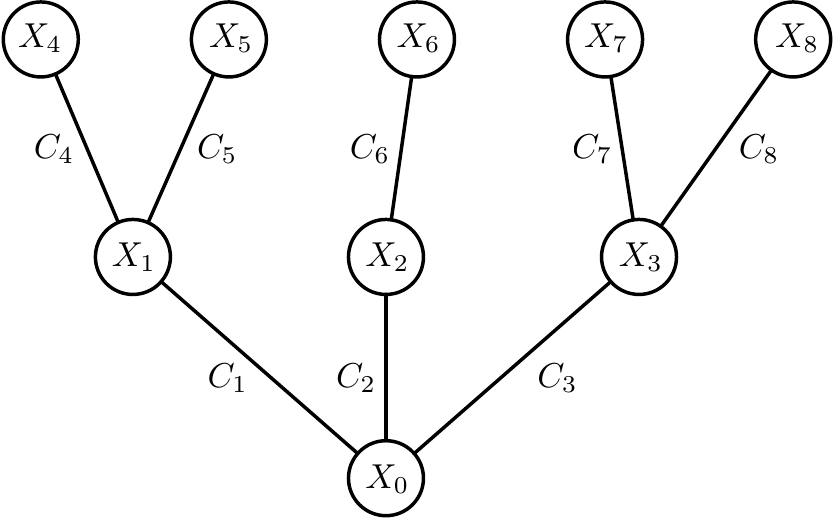}}
%\vspace{-2mm}
%\footnotesize
\caption{Fragment of the tree formed by the states.}\label{tree figure}
\vspace{-2mm}
\end{center}
%\vspace{-2mm}
\end{figure}
For any agent $a\in\mathcal{A}$ and any nodes $v,v'\in W$, we say that  $v\sim_a v'$ if all edges along  the unique simple path connecting nodes $v$ and $v'$ are labeled by agent $a$.

Throughout the rest of the article, for any nonempty sequence $w=x_1,\dots,x_n$, by $hd(w)$ we denote the element $x_n$  and by $w::y$ we denote sequence $x_1,\dots,x_n,y$.

The next lemma shows that the tree construction solves the challenge of the distributed knowledge discussed in the preamble to this section.
\begin{lemma}\label{transporter lemma}
If $w\sim_C w'$, then $\K_C\phi\in hd(w)$ iff $\K_C\phi\in hd(w')$.
\end{lemma}
\begin{proof}
Assumption $w\sim_C w'$ implies that each edge along the unique simple path between nodes $w$ and $w'$ is labeled with all agents in coalition $C$. Thus, it suffices to show that  $\K_C\phi\in hd(w)$ iff $\K_C\phi\in hd(w')$ for any two {\em adjacent} nodes along this path. Indeed, without loss of generality, let 
\begin{eqnarray*}
w&=&X_0,C_1,X_1,\dots,C_{n-1},X_{n-1} \\
w'&=&X_0,C_1,X_1,\dots,C_{n-1},X_{n-1},C_n,X_n.
\end{eqnarray*}
The assumption that the edge between $w$ and $w'$ is labeled with all agents in coalition $C$ implies that $C_n\subseteq C$. Next, we show that $\K_C\phi\in hd(w)$ iff $\K_C\phi\in hd(w')$.

\noindent$(\Rightarrow):$ Suppose that $\K_C\phi\in hd(w)=X_{n-1}$. Thus, $X_{n-1}\vdash\K_C\K_C\phi$ by Lemma~\ref{positive introspection lemma}. Hence, $X_{n-1}\vdash\K_{C_n}\K_C\phi$ by the Epistemic Monotonicity axiom and because $C_n\subseteq C$. Hence,
$\K_{C_n}\K_C\phi\in X_{n-1}$ because set $X_{n-1}$ is maximal. Then, $\K_C\phi\in X_{n}=X(w')$ by Definition~\ref{canonical state}.

\noindent$(\Leftarrow):$ Suppose that $\K_C\phi\notin hd(w)=X_{n-1}$. Thus, $\neg\K_C\phi\in X_{n-1}$ because set $X_{n-1}$ is maximal. Hence, $X_{n-1}\vdash\K_C\neg\K_C\phi$ by the Negative Introspection axiom. Hence, $X_{n-1}\vdash\K_{C_n}\neg\K_C\phi$ by the Epistemic Monotonicity axiom and because $C_n\subseteq C$. Hence,
$\K_{C_n}\neg\K_C\phi\in X_{n-1}$ because set $X_{n-1}$ is maximal. Then, $\neg\K_C\phi\in X_{n}$ by Definition~\ref{canonical state}. Therefore, $\K_C\phi\notin X_{n}=hd(w')$ because set $X_{n}$ is consistent.
\end{proof}

This defines the states of the canonical game $G(X_0)$ and the indistinguishability relations $\{\sim_a\}_{a\in\mathcal{A}}$ on these states. We now will define the domain of actions and the mechanism of the canonical game. 

%These definitions in the previous works on coalition know-how strategies~\cite{nt17aamas,nt17tark,nt18ai,nt18aaai,nt18aamas} have been using various modifications of the same basic idea: the domain of actions consists of all formulae in the language $\Phi$. The mechanism guarantees that if $\H_C\phi\in hd(w)$ and all members of the coalition $C$ choose action  $\phi$, then the game transitions into state $u$ such that  $\phi\in hd(u)$. To handle strategies with intelligence, we assume that agents choose not only a formula but also an integer number. If $[C]_B\phi\in hd(w)$, then in order to guarantee that $\phi\in hd(u)$, not only the members of the coalition $C$ must all choose action $\phi$, but they also must choose numbers that are larger than those chosen by the members of the coalition $B$. This way, coalition $C$ has a strategy to force $\phi$ only if it knows the numbers chosen by the coalition $B$. Note that this requires coalition $C$ to know the numbers, but not the formulae chosen by the members of the coalition $B$. The next two definition capture this informal idea.

\begin{definition}
$\Delta$ is a set of all pairs $(\phi, z)$ such that $\phi \in \Phi$ is a formula, and $z \in \mathbb{Z}$ is an integer number.
\end{definition}

If $u$ is a pair $(x,y)$, then by $pr_1(u)$ and $pr_2(u)$ we mean elements $x$ and $y$ respectively.

\begin{definition}\label{canonical M}
Mechanism $M$ is the set of triples $(w,\delta,u)$ such that for any formula $[C]_B\phi\in hd(w)$ if
\begin{enumerate}
    \item $pr_1(\delta(c))=\phi$, for each $c\in C$,
    \item $pr_2(\delta(b))< pr_2(\delta(c))$ for each $b\in B$ and each $c\in C$,
\end{enumerate}
then $\phi\in hd(u)$.
\end{definition}

Figure~\ref{mechanism figure} describes a Battle of the Atlantic inspired example that illustrates the definition of the canonical mechanism. 
\begin{figure}[ht]
\begin{center}
\begin{tabular}{|l|l|}  \hline &\\ [-4pt]
$hd(w)$ & 
$[\mbox{British}]_{\mbox{\footnotesize German}}(\mbox{saved}),
[\mbox{German}]_{\mbox{\footnotesize British}}(\mbox{not saved})
$\\ [4pt]
British & $(\mbox{saved},17)$\\ [4pt]
German & $(\mbox{not saved},23)$ \\ [4pt]
Russian & $(\mbox{saved},29)$\\[7pt] \hline & \\[-7pt]
Outcome & not saved \\[4pt] \hline
\end{tabular}
\caption{Battle of the Atlantic Mechanism.}\label{mechanism figure}
\vspace{-2mm}
\end{center}
%\vspace{-2mm}
\end{figure}
Here set $hd(w)$ contains formulae  $[\mbox{British}]_{\mbox{\footnotesize German}}(\mbox{saved})$ and $[\mbox{German}]_{\mbox{\footnotesize British}}(\mbox{not saved})$. Thus, the mechanism enables both British and Germans to achieve their goal as long as their have intelligence about the move of the other party. British, Germans, and Russians have chosen actions
$(\mbox{saved},17)$, $(\mbox{not saved},23)$, and $(\mbox{saved}, 29)$ respectively. Thus, $pr_2(\delta(\mbox{British}))=17<23=pr_2(\delta(\mbox{Germans}))$. Then, according to Definition~\ref{canonical M}, statement ``saved'' (short for "Convoy is saved'') will belong to set $hd(u)$, where $u$ is the outcome state of the game. Note that although $pr_2(\delta(\mbox{German}))=23<29=pr_2(\delta(\mbox{Russians}))$, the Russian action did not save the convoy because statement $[\mbox{Russian}]_{\mbox{\footnotesize German}}(\mbox{saved})$ does not belong to set $hd(w)$.

\begin{definition}\label{canonical pi}
$\pi(p)=\{w\in W\;|\; p\in hd(w)\}$.
\end{definition}

This concludes the definition of the canonical game $G(X_0)=(W,\{\sim_a\}_{a\in\mathcal{A}},$ $\Delta,M,\pi)$.
The next important milestone in the proof of the completeness is what sometimes is called ``truth'' lemma that connects syntax and semantics sides in the canonical game construction. In our case, this is Lemma~\ref{main induction lemma}. Before that lemma, however, we state and prove two auxiliary statements that will be used in the induction step of the proof of Lemma~\ref{main induction lemma}.

\begin{lemma}\label{u exists lemma}
If $\neg\K_C\phi\in hd(w)$, then there is a state $u\in W$ such that $w\sim_C u$ and $\neg\phi\in hd(u)$.
\end{lemma}
\begin{proof}
Consider set of formulae
$$
X=\{\neg\phi\}\cup\{\psi\;|\;\K_C\psi\in w\}.
$$
First, we prove that set $X$ is consistent. Suppose the opposite. Thus, there are formulae $\K_C\psi_1,\dots,\K_C\psi_n\in hd(w)$ such that
$$
\psi_1,\dots,\psi_n\vdash\phi.
$$
Hence, by applying $n$ times Lemma~\ref{deduction lemma},
$$
\vdash\psi_1\to(\psi_2\to\dots (\psi_n\to\phi)\dots).
$$
Then, by the Epistemic Necessitation inference rule,
$$
\vdash\K_C(\psi_1\to(\psi_2\to\dots (\psi_n\to\phi)\dots)).
$$
Thus, by the Distributivity axiom and the Modus Ponens inference rule,
$$
\vdash\K_C\psi_1\to\K_C(\psi_2\to\dots (\psi_n\to\phi)\dots)).
$$
Recall that $\K_C\psi_1\in hd(w)$ by the choice of formula $\K_C\psi_1$. Hence, by the Modus Ponens inference rule,
$$
hd(w)\vdash\K_C(\psi_2\to\dots (\psi_n\to\phi)\dots)).
$$
By repeating the previous step $n-1$ more times,
$$
hd(w)\vdash\K_C\phi.
$$
Hence, $\neg\K_C\phi\notin hd(w)$ due to the consistency of the set $hd(w)$. This contradicts the assumption of the lemma. Therefore, set $X$ is consistent. By Lemma~\ref{Lindenbaum's lemma}, there is a maximal consistent extension   $\hat{X}$ of the set $X$. Let $u$ be sequence $w::C::\hat{X}$. Note that $u\in W$ by Definition~\ref{canonical state} and the choice of set $X$, set $\hat{X}$ and sequence $u$. Furthermore, $w\sim_C u$ by the definition of relation $\sim_a$ on the set $W$. Finally, $\neg\phi\in X\subseteq\hat{X}=hd(u)$ again by the choice of set $X$, set $\hat{X}$ and sequence $u$.
\end{proof}

\begin{lemma}\label{beta exists lemma}
If $\neg[C]_B\phi\in hd(w)$, then there exists an action profile $\beta\in\Delta^B$ of coalition $B$ such that for each action profile $\gamma\in\Delta^C$ of coalition $C$ there is a complete action profile $\delta\in\Delta^\mathcal{A}$ and  states $w',u\in W$ such that $\beta=_B\delta$, $\gamma=_C\delta$, $w\sim_C w'$, $(w',\delta,u)\in M$, and $\neg\phi\in hd(u)$.
\end{lemma}
\begin{proof}
Let action profile $\beta\in \Delta^B$ of coalition $B$ be such that $\beta(b) = (\top, 0)$ for each agent $b\in B$. Consider any action profile $\gamma \in \Delta^C$ of coalition $C$. Choose an integer $z_0$ such that for each $a\in C$
\begin{equation}\label{choice of z0}
    pr_2(\gamma(a))< z_0.
\end{equation}
Such $z_0$ exists because coalition $C$ is a finite set of agents.
Define complete action profile $\delta\in \Delta^\mathcal{A}$ as follows
\begin{equation}\label{choice of delta}
\delta(a)=
\begin{cases}
\beta(a), & \mbox{if } a \in B, \\
\gamma(a), & \mbox{if } a \in C,\\
(\top, z_0),  & \mbox{otherwise.}
\end{cases}
\end{equation}
Note that $\neg[C]_B\phi\in hd(w)\subseteq\Phi$ by the assumption of the lemma. Thus, sets $B$ and $C$ are disjoint by Definition~\ref{Phi}. Thus, complete action profile $\delta$ is well-defined.

Consider set $X$ such that
\begin{eqnarray}
X&=&\{\neg\phi\}\cup \{\sigma\;|\; [\varnothing]_\varnothing\sigma\in hd(w)\} \cup \label{choice of X}\\
&&\{ \psi \;|\; [P]_Q \psi \in hd(w),P\neq\varnothing,  \forall p \in P (pr_1(\delta(p)) = \psi),\nonumber\\
&&\hspace{7mm}\forall q \in Q\,\forall p \in P (pr_2(\delta(q)) < pr_2(\delta(p)))\}\nonumber.
\end{eqnarray}
Next we show that set $X$ is consistent. Suppose the opposite. Thus,
\begin{equation}\label{sigma psi phi}
    \sigma_1,\dots,\sigma_m,\psi_1, \psi_2, \psi_3, \dots, \psi_n \vdash \phi
\end{equation}
for some formulae
\begin{equation}\label{choice of sigma}
[\varnothing]_\varnothing\sigma_1,\dots,[\varnothing]_\varnothing\sigma_m\in hd(w)
\end{equation}
 and some formulae
 \begin{equation}\label{choice of psi}
 [P_1]_{Q_1}\psi_1,\dots,[P_n]_{Q_n}\psi_n\in hd(w)
 \end{equation}
 such that
\begin{equation}\label{q<p equation}
    pr_2(\delta(q))<pr_2(\delta(p)),
\end{equation}
\begin{equation}\label{Pi nonempty}
    P_i\neq \varnothing,
\end{equation}
and
\begin{equation}\label{delta p}
    pr_1(\delta(p))=\psi_i
\end{equation}
for each $i\le n$, each $q\in Q_i$, and each $p\in P_i$. Without loss of generality, we can assume that formulae $\psi_1, \psi_2, \psi_3, \dots, \psi_n$ are distinct and none of them is equal to $\top$:
\begin{equation}\label{i neq j}
    \psi_i\neq \psi_j,
\end{equation}
\begin{equation}\label{psi_i not top}
    \psi_i\neq \top
\end{equation}
for each $i\le n$ and each $j\neq i$.
\begin{claim}\label{Ps disjoint}
Sets $P_1,\dots,P_n$ are pairwise disjoint.
\end{claim}
\begin{proof}
Consider any agent $a \in P_i \cap P_j$. Then, $\psi_i = pr_1(\delta(a)) = \psi_j$ by equation~(\ref{delta p}), which contradicts assumption~(\ref{i neq j}). 
\end{proof}

\begin{claim}\label{PiC claim}
$P_i\subseteq C$ for each $i\le n$.
\end{claim}
\begin{proof}
Consider any agent $a \in P_i$. Suppose that $a \notin C$. Then, $pr_1(\delta(a)) = \top$ by equation~(\ref{choice of delta}). At the same time, $pr_1(\delta(a))=\psi_i$ by equality~(\ref{delta p}) because $a\in P_i$. Hence, $\psi_i=\top$, which contradicts assumption~(\ref{psi_i not top}).
\end{proof}

\begin{claim}\label{QBC claim}
$Q_i\subseteq B\cup C$.
\end{claim}
\begin{proof}
Consider any agent $q\in Q_i$. Statement~(\ref{Pi nonempty}) implies that there is at least one agent $p\in P_i$. Then, $p\in C$ by Claim~\ref{PiC claim}. Thus, $pr_2(\delta(p))=pr_2(\gamma(p))<z_0$ due to equality~(\ref{choice of delta}) and inequality~(\ref{choice of z0}). Hence, $pr_2(\delta(q))<z_0$ by inequality~(\ref{q<p equation}). Therefore, $q\in B\cup C$ due to equality~(\ref{choice of delta}).
\end{proof}

For any nonempty finite set of agents $P\subseteq\mathcal{A}$ let 
\begin{equation}\label{rank definition}
    rank(P)=\min_{p\in P} pr_2(\delta(p)).
\end{equation}
Sets $P_1,\dots,P_n$ are nonempty by statement~(\ref{Pi nonempty}). Thus, $rank(P_i)$ is defined for each $i\le n$. Without loss of generality, we can assume that, see Figure~\ref{rank figure},
\begin{figure}[ht]
\begin{center}
%\vspace{-4mm}
\scalebox{0.7}{\includegraphics{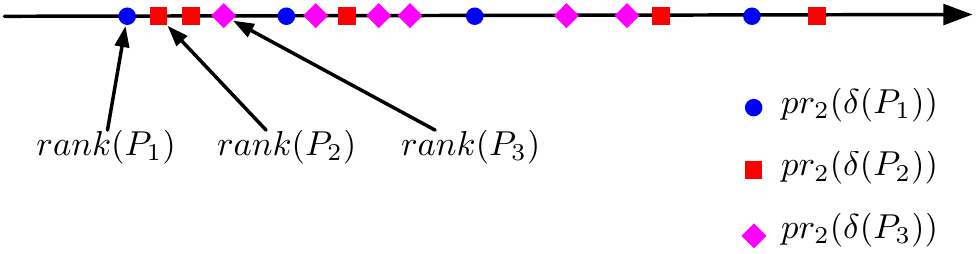}}
%\vspace{-2mm}
%\footnotesize
\caption{Sets $P_1,P_2,P_3,\dots$ and their ranks.}\label{rank figure}
\vspace{-2mm}
\end{center}
%\vspace{-2mm}
\end{figure}

\begin{equation}\label{rank order}
    rank(P_1)\le rank(P_2)\le \dots \le rank(P_n).
\end{equation}

\begin{claim}\label{QiPj claim}
Sets $Q_i$ and $P_j$ are disjoint for $1\le i\le j$.
\end{claim}
\begin{proof}
For any $q\in Q_i$ and any $p\in P_j$, by inequality~(\ref{q<p equation}) and definition~(\ref{rank definition}); assumption~(\ref{rank order}); and again definition~(\ref{rank definition}), 
$$
pr_2(\delta(q))<rank(P_i)\le rank(P_j)\le pr_2(\delta(p)).
$$
Therefore, sets $Q_i$ and $P_j$ are disjoint.
\end{proof}

Let
\begin{equation}\label{choice of R}
    R=C\setminus(P_1\cup\dots\cup P_n).
\end{equation}

\begin{claim}\label{BRPs disjoint}
Sets $B$, $R$, $P_1$, \dots, $P_n$ are pairwise disjoint. 
\end{claim}
\begin{proof}
Assumption $\neg[C]_B\phi\in hd(w)$ of the lemma implies that $\neg[C]_B\phi\in \Phi$. Thus, sets $B$ and $C$ are disjoint by Definition~\ref{Phi}. Hence, sets $B$ and $R$ are disjoint because of equation~(\ref{choice of R}) and set $B$ is disjoint with each of sets $P_1,\dots, P_n$ by Claim~\ref{PiC claim} and because sets $B$ and $C$ are disjoint. Also, set $R$ is disjoint with each of sets $P_1,\dots, P_n$ by equation~(\ref{choice of R}). Finally, sets $P_1,\dots, P_n$ are pairwise disjoint by Claim~\ref{Ps disjoint}.
\end{proof}

\begin{claim}\label{QiBRPs claim}
$Q_i\subseteq B\cup R\cup P_1\cup \dots\cup P_{i-1}$ for  $1\le i \le n$.
\end{claim}
\begin{proof}
Consider any agent $q\in Q_i$ such that $q\notin B\cup R$. It suffices to show that $q\in P_1\cup \dots\cup P_{i-1}$. Indeed, assumptions $q\in Q_i$ and $q\notin B$ imply that $q\in C$ by Claim~\ref{QBC claim}. Thus, $q\in P_1\cup\dots\cup P_n$ by the assumption $q\notin R$ and the definition of set $R$. Therefore, $q\in P_1\cup\dots\cup P_{i-1}$ by Claim~\ref{QiPj claim} because $q\in Q_i$.
\end{proof}

Let us now return to the proof of the lemma. Statement~(\ref{sigma psi phi}), by the Lemma~\ref{deduction lemma} applied $m+n$ times, implies that  
$$
\vdash \sigma_1 \to (\dots(\sigma_m\to (\psi_1 \to \dots ( \psi_n \to \phi) \dots ))\dots).
$$
Hence, by the Strategic Necessitation inference rule,
$$
\vdash [\varnothing]_{\varnothing}(\sigma_1 \to (\dots(\sigma_m\to (\psi_1 \to \dots ( \psi_n \to \phi) \dots ))\dots)).
$$
Thus, by the Cooperation axiom (where $B=C=D=\varnothing$) and the Modus Ponens inference rule,
$$
\vdash [\varnothing]_\varnothing\sigma_1 \to [\varnothing]_{\varnothing}(\sigma_2 \to (\dots(\sigma_m\to (\psi_1 \to \dots ( \psi_n \to \phi) \dots ))\dots)).
$$
Then, by the Modus Ponens inference rule and assumption~(\ref{choice of sigma}),
$$
hd(w) \vdash [\varnothing]_{\varnothing}(\sigma_2 \to (\dots(\sigma_m\to (\psi_1 \to \dots ( \psi_n \to \phi) \dots ))\dots)).
$$
By repeating the previous step $m-1$ times,
$$
hd(w) \vdash [\varnothing]_{\varnothing} (\psi_1 \to (\psi_2 \to (\psi_3 \to \dots ( \psi_n \to \phi) \dots ))). 
$$
Hence, by the Intelligence Monotonicity axiom and the Modus Ponens inference rule,
$$
hd(w) \vdash [\varnothing]_{B} (\psi_1 \to (\psi_2 \to (\psi_3 \to \dots ( \psi_n \to \phi) \dots ))). 
$$
Thus, by Lemma~\ref{subscript monotonicity lemma}, Claim~\ref{BRPs disjoint}, and the Modus Ponens inference rule,
$$
hd(w) \vdash [R]_{B} (\psi_1 \to (\psi_2 \to (\psi_3 \to \dots ( \psi_n \to \phi) \dots ))). 
$$
Then, by the Cooperation axiom, Claim~\ref{BRPs disjoint}, and the Modus Ponens inference rule,
$$
hd(w) \vdash [{P_1}]_{B,R} \psi_1 \to [{R, P_1}]_{B}(\psi_2 \to (\psi_3 \to \dots (\psi_n \to \phi) \dots)).
$$
At the same time, recall that $[{P_1}]_{Q_1}\psi_1\in hd(w)$ by assumption~(\ref{choice of psi}). Thus, $hd(w)\vdash [{P_1}]_{B,R}\psi_1$ by the Intelligence Monotonicity axiom and because $Q_1\subseteq B\cup R$ due to Claim~\ref{QiBRPs claim}. Hence, by the Modus Ponens inference rule,
$$
hd(w) \vdash [{R, P_1}]_{B}(\psi_2 \to (\psi_3 \to \dots (\psi_n \to \phi) \dots)).
$$
Then, by the Cooperation axiom, Claim~\ref{BRPs disjoint}, and the Modus Ponens inference rule,
$$
hd(w) \vdash [{P_2}]_{B, R, P_1}\psi_2 \to [{R, P_1, P_2}]_B(\psi_3 \to \dots (\psi_n \to \phi) \dots)).
$$
At the same time, recall that $[{P_2}]_{Q_2}\psi_2\in hd(w)$ by assumption~(\ref{choice of psi}). Thus, $hd(w)\vdash [{P_2}]_{B,R,P_1}\psi_2$ by the Intelligence Monotonicity axiom and because $Q_2\subseteq B,R,P_1$ due to Claim~\ref{QiBRPs claim}. Hence, by the Modus Ponens inference rule,
$$
hd(w) \vdash [{R, P_1, P_2}]_B(\psi_3 \to \dots (\psi_n \to \phi) \dots)).
$$
By repeating the previous step $n-2$ more times,
% $$
% w \vdash \H^B_{R, P_1, P_2, \dots, P_{n-1}} (\psi_n \to \phi)
% $$
% $$
% w \vdash \H^{B, R, P_1, P_2, \dots, P_{n-1}}_{P_n} \psi_n \to  \H^B_{R, P_1, P_2, \dots, P_n} \phi$$
$$
 hd(w) \vdash [{R, P_1, P_2, \dots, P_n}]_B\phi.
$$
Equation~(\ref{choice of R}) implies that $R\subseteq C$. Thus, $R, P_1, P_2, \dots, P_n\subseteq C$ by Claim~\ref{PiC claim}. Then,
$ hd(w) \vdash [C]_B \phi$
by Lemma~\ref{subscript monotonicity lemma}, which contradicts the assumption $\neg [C]_B \phi\in hd(w)$ and the consistency of set $hd(w)$. Therefore, set $X$, as defined by equation~(\ref{choice of X}), is consistent. By Lemma~\ref{Lindenbaum's lemma}, there exists a  maximal consistent extension $\hat{X}$ of the set $X$. 

Let $w'$ be state $w$ and $u$ to be the sequence $w::\varnothing::\hat{X}$.

\begin{claim}
$u\in W$.
\end{claim}
\begin{proof}
By Definition~\ref{canonical state}, it suffices to prove that $$\{\phi\in\Phi\;|\;\K_\varnothing\phi\in hd(w)\}\subseteq hd(u).$$ Indeed, let $\K_\varnothing\phi\in hd(w)$. Thus, $hd(w)\vdash [\varnothing]_\varnothing\phi$ by the Empty Coalition axiom. Hence, $[\varnothing]_\varnothing\phi\in hd(w)$ due to the maximality of the set $hd(w)$. Then, $\phi\in X$ by equation~(\ref{choice of X}). Thus, $\phi\in \hat{X}$ by the choice of set $\hat{X}$. Therefore, $\phi\in hd(u)$ by the choice of sequence $u$.
\end{proof}

Note that $\beta=_B\delta$ because of equation~(\ref{choice of delta}) and the assumption $\beta\in\Delta^B$ of the lemma. Similarly, $\gamma=_C\delta$. Also, $w\sim_C w'$ by Lemma~\ref{sim C is equivalence relation lemma} because $w'=w$. Additionally, $\neg\phi\in hd(u)$ because, due to equation~(\ref{choice of X}), we have $\neg\phi\in X\subseteq\hat{X}=hd(u)$.

Since $w=w'$, to finish the proof of the lemma, we need to show that $(w,\delta,u)\in M$. Consider any formula $[P]_Q\psi\in hd(w)$ such that $pr_1(\delta(p))=\psi$ for each agent $p\in P$ and $pr_2(\delta(q))<pr_2(\delta(p))$ for each agent $q\in Q$ and each agent $p\in P$. By Definition~\ref{canonical M}, it suffices to prove that $\psi\in hd(u)$. We consider the following tow cases separately.

\noindent{\bf Case I}: $P\neq \varnothing$. Thus, $\psi\in X$ by equation~(\ref{choice of X}). Therefore, $\psi\in X\subseteq\hat{X}=hd(u)$. 

\noindent{\bf Case II}: $P= \varnothing$. Then, assumption $[P]_Q\psi\in hd(w)$ can be rewritten as $[\varnothing]_Q\psi\in hd(w)$. Thus, $hd(w)\vdash [\varnothing]_\varnothing\psi$ by the None to Analyze axiom. Hence, $[\varnothing]_\varnothing\psi\in hd(w)$ because of the maximality of the set $hd(w)$. Thus, $\psi\in X$ by equation~(\ref{choice of X}). Therefore, $\psi\in X\subseteq\hat{X}=hd(u)$. 

This concludes the proof of Lemma~\ref{beta exists lemma}.
\end{proof}

We are now ready to state and to prove the main induction lemma of the proof of the completeness, which sometimes also is referred to as truth lemma.

\begin{lemma}\label{main induction lemma}
$w\Vdash \phi$ iff $\phi\in hd(w)$ for each formula $\phi\in\Phi$.
\end{lemma}
\begin{proof}
We prove the lemma by the induction on the structural complexity of formula $\phi$. The case follows from Definition~\ref{sat} and Definition~\ref{canonical pi}. The case when formula $\phi$ is a negation or an implication follows from Definition~\ref{sat} and the assumption that set $hd(w)$ is a maximal consistent set of formulae in the standard way.

Suppose that formula $\phi$ has the form $\K_C\psi$.

\noindent $(\Rightarrow)$: Suppose that $\K_C\psi\notin hd(w)$. Thus, $\neg\K_C\psi\in hd(w)$ due to the maximality of the set $hd(u)$. Hence, by Lemma~\ref{u exists lemma},  there is a state $u\in W$ such that $w\sim_C u$ and $\neg\psi\in hd(u)$. Then, $\psi\notin hd(u)$ due to the consistency of the set $hd(u)$. Thus, $u\nVdash\psi$ by the induction hypothesis. Therefore, $w\nVdash\K_C\psi$ by Definition~\ref{sat}. 

\noindent $(\Leftarrow)$: Assume that $\K_C\psi\in hd(w)$. Consider any state $u\in W$ such that $w\sim_C u$. By Definition~\ref{sat}, it suffices to show that $u\Vdash \psi$. Indeed, by Lemma~\ref{transporter lemma}, assumptions $\K_C\psi\in hd(w)$ and $w\sim_C u$ imply that $\K_C\psi\in hd(u)$. Hence, $hd(u)\vdash\psi$ by the Truth axiom and the Modus Ponens inference rule. Thus, $\psi\in hd(u)$ due to the maximality of the set $hd(u)$. Therefore, $u\Vdash u$ by the induction hypothesis.

Suppose that formula $\phi$ has the form $[C]_B\psi$.

\noindent $(\Rightarrow)$: Assume that $[C]_B\psi\notin hd(w)$. Hence, $\neg[C]_B\psi\in hd(w)$ due to the maximality of the set $hd(w)$. Thus, by Lemma~\ref{beta exists lemma}, there exists an action profile $\beta\in\Delta^B$ of coalition $B$ such that for each action profile $\gamma\in\Delta^C$ of coalition $C$ there is a complete action profile $\delta\in\Delta^\mathcal{A}$ and  states $w',u\in W$ such that $\beta=_B\delta$, $\gamma=_C\delta$, $w\sim_C w'$, $(w',\delta,u)\in M$, and $\neg\psi\in hd(u)$. Note that $\neg\psi\in hd(u)$ implies $\psi\notin hd(u)$ due to the consistency of the set $hd(u)$, which, in term implies $u\nVdash\psi$ by the induction hypothesis. 

Thus, there exists an action profile $\beta\in\Delta^B$ of coalition $B$ such that for each action profile $\gamma\in\Delta^C$ of coalition $C$ there is a complete action profile $\delta\in\Delta^\mathcal{A}$ and  states $w',u\in W$ such that $\beta=_B\delta$, $\gamma=_C\delta$, $w\sim_C w'$, $(w',\delta,u)\in M$, and $u\nVdash\psi$. Therefore, $w\nVdash[C]_B\psi$ by Definition~\ref{sat}. 

\noindent $(\Leftarrow)$: Assume that $[C]_B\psi\in hd(w)$. Consider any action profile $\beta\in \Delta^B$ of coalition $B$. Set $B$ is finite by Definition~\ref{Phi}. Let $z_0$ be any integer number such that $pr_2(\beta(b))<z_0$ for each agent $b\in B$. Define action profile $\gamma\in\Delta^C$ as $\gamma(c)=(\phi,z_0)$ for each $c\in C$. Consider any complete action profile $\delta\in \Delta^\mathcal{A}$, any state $w'\in W$, and any state $u\in W$ such that $\beta=_B\delta$, $\gamma=_C\delta$, $w\sim_C w'$, and $(w',\delta,u)\in M$. By Definition~\ref{sat}, it suffices to show that $u\Vdash\psi$.

Assumption $[C]_B\psi\in hd(w)$ implies that $hd(w)\vdash \K_C[C]_B\psi$ by the Strategic Introspection axiom and the Modus Ponens inference rule. Thus, $\K_C[C]_B\psi\in hd(w)$ by the maximality of the set $hd(w)$. Hence, $\K_C[C]_B\psi\in hd(w')$ by the assumption $w\sim_C w'$ and Lemma~\ref{transporter lemma}. Then, $hd(w')\vdash [C]_B\psi$ by the Truth axiom and the Modus Ponens inference rule. Hence, $[C]_B\psi\in hd(w')$ due to the maximality of the set $hd(w')$.

By the choice of action profile $\gamma$ and the assumption $\gamma=_C\delta$, for each $c\in C$, we have $pr_1(\delta(c))=pr_1(\gamma(c))=\phi$. At the same time, by the assumption $\beta=_B\delta$, the choice of integer $z_0$, the choice of action profile $\gamma$, and the assumption $\gamma=_C\delta$,
$$
pr_2(\delta(b))=pr_2(\beta(b))<z_0=pr_2(\gamma(c))=pr_2(\delta(c))
$$
for each agent $b\in B$ and each agent $c\in C$.

Thus, $\psi\in hd(u)$ by Definition~\ref{canonical M} and the assumption $(w',\delta,u)\in M$. Therefore, $u\Vdash\psi$ by the induction hypothesis.
\end{proof}

Now we are ready to state and to prove the strong completeness of our logical system.

\begin{theorem}
If $Y\nvdash\phi$, then there is a state $w$ of a game such that $w\Vdash\chi$ for each $\chi\in Y$ and $w\nVdash\phi$. 
\end{theorem}
\begin{proof}
Suppose that $Y\nvdash\phi$. By Lemma~\ref{Lindenbaum's lemma}, there exists a  maximal consistent set of formulae $X_0$ such that  $Y\cup\{\neg\phi\}\subseteq X_0$. Let $w$ be the single-element sequence $X_0$. By Definition~\ref{canonical state}, sequence $w$ is a state of the canonical game $G(X_0)$. Note that $\chi\in X_0=hd(w)$ for each formula $\chi\in Y$ and $\neg\phi\in X_0=hd(w)$ by the choice of set $X_0$ and the choice of sequence $w$. Thus, $w\Vdash\chi$ for each $\chi\in Y$ and  $w\Vdash\neg\phi$ by Lemma~\ref{main induction lemma}. Therefore, $w\nVdash\phi$ by Definition~\ref{sat}.
\end{proof}

\section{Conclusion}\label{concusion section}

In this article we proposed the notion of a strategy with intelligence and gave a sound and complete axiomatization of a bimodal logic that describes the interplay between strategic power with intelligence and the distributed knowledge modalities in the setting of strategic games with imperfect information. A natural question is decidability of the proposed logical system. Unfortunately the standard filtration technique~\cite{g72jpl} can not be easily applied here to produce a finite model. Indeed, it is crucial for the proof of the completeness, see Definition~\ref{canonical M} that for each action there is another action with a higher value of the second component. Thus, for the proposed construction to work, the domain of choices must be infinite. One perhaps might be able to overcome this by changing the second component of the action from an infinite linear ordered set to a finite circularly ``ordered'' set as in rock-article-scissors game.  

Another possible extension of this work is to consider modality $\K_C^B$ that captures the knowledge of coalition $C$ after coalition $B$ disclosed to $C$ the actions that it intends to take.

%%%%%%%%%%%%%%%%%%%%%%%%%%%%%%%%%%%%%%%%%%%%%%%%%%%%%%%%%%%%%%%%%%%%%%%%%%%%%%%%%%%%%%%%%%%%%%%%%%%%%%%%%
%% bibliography: see CFP for number of permitted pages

\bibliographystyle{elsarticle-num}
\bibliography{sp}  % put name of your .bib file here

\vfill

\pagebreak

\end{document}